\documentclass[10pt,journal,final,letterpaper,twocolumn,twoside]{IEEEtran}
\usepackage{color}
\usepackage{epsfig}
\usepackage{cite}

\newcommand{\Tr}{\mathrm{Tr}}

\usepackage{subfigure}
\usepackage{tikz}
\usepackage{tkz-berge}
\usepackage{pgfplots}

\usepackage{graphicx,color, amsmath,amssymb}
\usepackage{
            theorem,%
            amsfonts,%
            pifont}%

\def\yg{\mathbf{y}}
\def\dg{\mathbf{d}}

\def\Mg{\mathbf{M}}
\def\Ng{\mathbf{N}}
\def\Pg{\mathbf{P}}
\def\Ig{\mathbf{I}}
\def\Qg{\mathbf{Q}}
\def\Rg{\mathbf{R}}
\def\Fg{\mathbf{F}}
\def\Cg{\mathbf{C}}

\newcommand{\mR}{\mathbf R}

\newcommand{\beq}{\begin{equation}}
\newcommand{\eeq}{\end{equation}}
\newcommand{\be}{\begin{equation}}
\newcommand{\ee}{\end{equation}}
\newcommand{\bea}{\begin{eqnarray}}
\newcommand{\eea}{\end{eqnarray}}
\newcommand{\br}{\begin{eqnarray}}
\newcommand{\er}{\end{eqnarray}}
\newcommand{\brs}{\begin{eqnarray*}}
\newcommand{\ers}{\end{eqnarray*}}
\newcommand{\ba}{\begin{array}}
\newcommand{\ea}{\end{array}}

\newcommand{\al}{\alpha}

\newcommand{\siN}{\sum_{i=1}^N}
\newcommand{\sjn}{\sum_{j=1}^n}

\newcommand{\dii}{\dg_i\dg_i^T}

\newcommand{\cii}{\yg_i\yg_i^T}
\newcommand{\ciMv}{\yg_i^T\Mg^{-1}\yg_i}

\newcommand{\bed}{\begin{description}}
\newcommand{\eed}{\end{description}}
\newcommand{\iy}{\infty}

\def\al{\alpha}

\def\EOP{\ \hfill \rule{0.5em}{0.5em} }

\newenvironment{proof}[1][\!]{\noindent {\em Proof of #1. }}{\hfill $\blacksquare$ \vskip 3pt}

\newtheorem{thm}{Theorem}[section]
\newtheorem{lemma}[thm]{Lemma}

\begin{theorembodyfont}{\rmfamily}
\newtheorem{rmk}[thm]{Remark}

\end{theorembodyfont}

\newtheorem{deff}[thm]{Definition}

\begin{document}

\title{On the convergence of Maronna's $M$-estimators of scatter}

\author{Yacine Chitour, Romain Couillet~\IEEEmembership{Member,~IEEE,} and Fr\'ed\'eric Pascal~\IEEEmembership{Senior Member,~IEEE}\thanks{Chitour is with Laboratoire des Signaux et Syst\`emes at Sup\'elec, 91192 Gif s/Yvette, France and Universit\'e Paris Sud, Orsay, France {\tt yacine.chitour@lss.supelec.fr}. Couillet is with the Telecommunications Department at Sup\'elec {\tt romain.couillet@supelec.fr}. Pascal is with the SONDRA laboratory at Sup\'elec {\tt frederic.pascal@supelec.fr}. Chitour and Pascal's works were partially supported by the iCODE institute, research project of the Idex Paris-Saclay, while Couillet's work is funded by ERC--MORE EC--120133.}}

\maketitle

\begin{abstract}
In this paper, {we propose an alternative proof for the uniqueness} of Maronna's $M$-estimator of scatter \cite{maronna76robust} for $N$ vector observations $\yg_1,\ldots,\yg_N\in\mR^m$ under a mild constraint of linear independence of any subset of $m$ of these vectors. This entails in particular almost sure uniqueness for random vectors $\yg_i$ with a density as long as $N>m$. {This  approach allows to establish further relations that demonstrate that a properly normalized Tyler's $M$-estimator of scatter \cite{tyler1987distribution} can be considered as a limit of Maronna's $M$-estimator. More precisely, the contribution is to show that each $M$-estimator, verifying some mild conditions, converges towards a particular Tyler's $M$-estimator.} These results find important implications in recent works on the large dimensional (random matrix) regime of robust $M$-estimation.
\end{abstract}

\section{Introduction}
Subsequent to Huber's introduction of robust statistics in \cite{huber1964robust}, Maronna proposed in \cite{maronna76robust} a class of robust estimates for scatter matrices defined as the solution of an implicit equation. In \cite{maronna76robust}, the existence and uniqueness of such a solution are proved, under conditions involving both the ratio $c_N:=m/N$ of the population dimension $m$ and the sample size $N$, and the parametrization of the estimate. {This constraint was largely relaxed in \cite{kent1991redescending, zhang2013multivariate}}. With the recent renewed interest in robust $M$-estimation under the random matrix regime $N,m\to\infty$ with $c_N\to c_{\infty}\in(0,1)$ \cite{couillet2013robust,couillet2013therandom,zhang2014marchenko,soloveychik2014nonasymptotic}, {alternative proofs of existence and uniqueness have appeared motivated by this assumption of large $m$.} While Maronna's original results are valid for any (well-behaved) set of samples satisfying the condition on $c_N$, the results in e.g. \cite{couillet2013robust} are expressed in probabilistic terms and are only valid for all large $m,N$. 

Based on the ideas from \cite{pascal2008covariance,chitour2008exact,pascal2014generalized}, the present article {proposes an alternative proof to \cite{kent1991redescending} to show existence and uniqueness for all well-behaved set of samples with a known location parameter and for any  $c_N\in(0,1)$.} {More importantly, by a proper parametrization of the weight function appearing in Maronna's estimator, we prove that some sequences of Maronna's $M$-estimators converge to a unique Tyler's distribution-free $M$-estimator of scatter \cite{tyler1987distribution}. This result is a novel property of the Tyler's $M$-estimators, rigorously proved in this work.} 
{This completes the recent result (Theorem 1 of \cite{ollila2012distribution-free}) stating that the Tyler's $M$-estimator is the Maximum Likelihood estimator (MLE) of the scatter for various complex elliptically symmetric (CES) distributions as well as for the angular central Gaussian (ACG) distributions \cite{ollila2012complex}.}

The paper is organized as follows: Section II presents our main results as well as Monte-Carlo simulations that corroborate our theoretical claims, the proofs of which are provided in Section III. Section IV draws some conclusions and perspectives of this work.

\section{Notations and statement of the results}
Let $\mR_+$ (resp. $\mR_+^*$)  be the (resp. strictly) positive real line. We use $M_m(\mR)$ and ${\rm Sym}_m$ to denote the vector space of $m\times m$ matrices with real entries and the linear subspace of $M_m(\mR)$ made of the symmetric matrices, respectively. We also use ${\rm Sym}_m^+$ and ${\rm PSD}_m$ to denote the non trivial cones in $M_m(\mR)$ of the non negative symmetric matrices and of the symmetric positive definite matrices, respectively. Also, $(\cdot)^T$ stands for the transpose, $\Tr(\cdot)$ and $\det(\cdot)$ for the trace and the determinant. On $M_m(\mR)$, we use the inner product defined by the Frobenius norm $\Vert \mathbf A\Vert=\sqrt{\Tr(\mathbf A\mathbf A^T)}$. We also use $\leq$ to denote the partial order on ${\rm Sym}_m$ and $\Ig_m$ the $m\times m$ identity matrix. 
Functions of two non negative real variables $(t,x)$ will be considered. If $f$ is such a function, we use $f_t$, $f_x$, $f_{tx}$,~$\ldots$ to denote (when defined) the partial derivatives of $f$ with respect to $t$ and/or $x$.

\begin{deff}
 A family $(\yg_i)_{1\leq i\leq N}$ of vectors in $\mR^m$ is {\it admissible} if 
\bed
\item[$(C1)$] for $1\leq i\leq N$, $\Vert \yg_i\Vert =1$;
\item[$(C2)$] {the vectors in any subset of size $m$ of $\{\yg_1,\cdots, \yg_N\}$ are linearly independent} 
\eed
\end{deff}

{This definition straightforwardly implies that if $(\yg_i)_{1\leq i\leq N}$ is an admissible family of vectors in $\mR^m$ and if $m$ vectors (say) $\yg_1,\cdots,\yg_m$ which are then linearly independent by $(C2)$ are fixed, for $m+1\leq l\leq N$, we can write $\yg_l=\sjn \gamma_{lj}\yg_j$. Then, $\gamma_{lj}\neq 0$ for every $1\leq j\leq m$ and $m+1\leq l\leq N$.}\\

{Let us now} consider maps $u:(\mR_+^*)^2\rightarrow\mR_+$ of class $C^1$ satisfying:
\bed
\item[$(U1)$] $u(t,\cdot)$ is strictly decreasing;
\item[$(U2)$] for every $t>0$, $v(t,x):=x\mapsto xu(t,x)$ is increasing on $\mR_+$ and 
$l_t:=\sup_{x\geq 0}v(t,x)>m$;
\eed
We furthermore define, for every $x>0$, $u(0,x)=\frac m{x}$.
Note that, by continuity of $u$, $\forall x>0$, $\lim_{t\rightarrow 0^+}v(t,x)=m$. 
Also, according to $(U1)$ and $(U2)$, for each $t,x>0$, 
\beq\label{eq:v0}
v(t,x)=m+tv_1(x)+tw(t,x),
\eeq
with $v_1(\cdot):=v_t(0,\cdot)$ and $\forall x>0,\lim_{t\rightarrow 0}w(t,x)=0$. By a simple computation, one has that $v_1$ is a nondecreasing function on $\mR_+^*$. 

For further use, we introduce the following additional notation. Let $x_t>0$ be the unique positive number such that, $\forall t>0,~v(t,x_t)=x_tu(t,x_t)=m.$


We further consider the following assumption
\begin{align*}
(U3)\quad
\left\{
\begin{array}{ll}
v_x:=dv/dx>0 \\
v_1\hbox{ is increasing} \\ 
0<\liminf_{t\to 0} x_t\leq \limsup_{t\to 0} x_t<\iy. 
\end{array}
\right.
\end{align*}
If the latter occurs and $u$ is of class $C^2$, then $w(t,x)=tw_1(x)+o(t)$, with $w_1(\cdot):=w_t(0,\cdot)$ continuous on $(\mR_+^*)^2$, the convergence in $(U2)$ is uniform in $x$ on any compact of $\mR_+^*$ and $x_t$ converges to the unique solution $x_0$ of $v_1(x)=0$.

We use $\bar{u}(t,x)$ to denote the particular function 
\beq\label{model}
\bar{u}(t,x)=\frac{m(1+t)}{x+t}
\eeq
which is analytic on every compact of $(\mR_+)^2\setminus \{(0,0)\}$. Moreover, $\bar{l}_t=m(1+t)$, $\bar{v_1}(x)=m(1-\frac 1 {x})$ and $\bar w(t,x)=\frac{-mt}{t+x}$. 

The objective of the work is to study the solutions of the equation given, for all $t>0$, by
\begin{equation*}
({\rm Eq})_t\qquad \Mg=\frac{1}{N}\sum_{i=1}^Nu(t,\ciMv)\cii.
\end{equation*}
and to characterize them in the limit where $t\to 0$.
Taking into account our definitions, if a solution to $({\rm Eq})_t$ exists, it must belong to ${\rm PSD}_m$.

{Remark that the condition M of \cite{kent1991redescending} also imposes a ``strictly'' increasing $v$ which excludes e.g. the Huber $M$-estimator.}



To state our results, we need to consider the set of solutions of the equation $({\rm Eq})_0$ {(that defines the Tyler's $M$-estimator)} given by
\begin{equation*}
({\rm Eq})_0\qquad \Mg=\frac{m}{N}\sum_{i=1}^N\frac1{\ciMv}\cii.
\end{equation*}
Recall from \cite{pascal2008covariance} that the set of solutions of $({\rm Eq})_0$ is the half-line $\mR_+^*\,\Pg$ in ${\rm PSD}_m$, where $\Pg$ is the unique solution of $({\rm Eq})_0$ with $\Tr(\Pg)=m$.\\

Our main result is the following theorem.
\begin{thm}\label{th1}
	Let $(\yg_i)_{1\leq i\leq N}$ be an admissible family of vectors in $\mR^m$ and $u:(\mR_+)^2\setminus \{(0,0)\}\to \mR_+$ be a $C^1$ function verifying $(U1)$--$(U2)$.
Then,
\bed
\item[$(A)$] $\forall t>0$, $({\rm Eq})_t$ admits a unique solution, $\Mg(t)$.
\item[$(B)$] If, furthermore, $u$ is $C^2$ and satisfies $(U3)$, then the mapping $t\mapsto \Mg(t)$ is continuous and $\lim_{t\to 0}\Mg(t)=\Mg_0$ the solution of $({\rm Eq})_0$ given by $\Mg_0=\xi_u \Pg$ with
$\xi_u>0$ unique solution to
\beq\label{eq:xi0}
\siN v_1\left(\frac{\yg_i^T\Pg^{-1}\yg_i}\xi\right)=0.
\eeq
In particular, for $u=\bar{u}$, $\Mg_0=\Pg$, i.e., $\xi_{\bar u}=1$.
\eed
\end{thm}

\begin{proof}[Theorem \ref{th1}]
The proof is postponed in the next section.
\end{proof}

\begin{rmk}~

\begin{enumerate}
\item The interest of Theorem \ref{th1}, in addition to providing an alternative proof for the existence and uniqueness, lies in the convergence of all $M$-estimators to a Tyler's $M$-estimator. This limit can be different (by a scale factor) from one $M$-estimator to another. While this result was expected, this paper rigorously proves it.\\

\item Moreover, the theorem provides a way of understanding why the Tyler's estimator is the outmost robust\footnote{Here the robustness has to be understood as the classical property considered in the robust estimation theory literature, see e.g. \cite{hampel1986robust}} $M$-estimator. Indeed, considering a ML approach, the weight function $u(t,x)$ is derived from the observations probability density function (PDF) and in such a case, $t \to 0$ means that the underlying distribution becomes more and more heavy-tailed. For instance, considering $t$ as the exponent parameter of a Generalized Gaussian distribution or of a W-distribution, the smaller the value of $t>0$ is, the heavier-tailed is the distribution. This is also the case for the degree of freedom of a Student-t distribution or the shape parameters of a K-distribution or of a Compound-Gaussian with inverse Gaussian texture (see \cite{ollila2012complex} for more details).  In all these cases, the MLEs satisfy the assumptions of Theorem \ref{th1} (at least for small values of $t$) and should be more robust when the distribution is heavier-tailed. To summarize, this result theoretically motivates the use of the Tyler's estimator, since it will perform similarly as MLEs in heavy-tailed distribution contexts.\\
\end{enumerate}
\end{rmk}

\begin{figure}[!h]
\begin{center}
{\label{err-cv=0.99}\begin{tikzpicture}[font=\footnotesize,scale=.98]
\pgfplotsset{every axis/.append style={mark options=solid, mark size=2.5pt}}
\pgfplotsset{every axis legend/.append style={fill=white,cells={anchor=west},at={(0.02,0.98)},anchor=north west}} \tikzstyle{every axis y label}+=[yshift=-10pt]
\tikzstyle{every axis x label}+=[yshift=5pt]
\tikzstyle{dashed dotted}=[dash pattern=on 1pt off 4pt on 6pt off 4pt]

\begin{axis}[xlabel={$t$},ylabel={$C(t)$}, xmin=0,xmax=1.001,ymin=0,ymax=430]
\addplot[mark=star,smooth,red,line width=.5pt] plot coordinates {
(0.001,0.00888409)(0.051,10.7472)(0.101,25.4708)(0.151,38.6432)(0.201,49.8004)(0.251,59.187)(0.301,67.1252)(0.351,73.898)(0.401,79.7324)(0.451,84.807)(0.501,89.2614)(0.551,93.2052)(0.601,96.7251)(0.651,99.8898)(0.701,102.755)(0.751,105.365)(0.801,107.757)(0.851,109.96)(0.901,111.999)(0.951,113.895)(1.001,115.665)
};
\addplot[mark=diamond,smooth,blue,line width=.5pt] plot coordinates {
(0.001,0.0121789)(0.051,14.5589)(0.101,34.1226)(0.151,51.3669)(0.201,65.8162)(0.251,77.8664)(0.301,87.9791)(0.351,96.5449)(0.401,103.872)(0.451,110.201)(0.501,115.716)(0.551,120.565)(0.601,124.86)(0.651,128.693)(0.701,132.137)(0.751,135.251)(0.801,138.082)(0.851,140.67)(0.901,143.048)(0.951,145.242)(1.001,147.275)
};
\addplot[mark=triangle,smooth,green,line width=.5pt] plot coordinates {
(0.001,0.0289981)(0.051,37.5272)(0.101,90.7745)(0.151,138.733)(0.201,179.089)(0.251,212.568)(0.301,240.341)(0.351,263.489)(0.401,282.897)(0.451,299.273)(0.501,313.173)(0.551,325.039)(0.601,335.22)(0.651,343.997)(0.701,351.596)(0.751,358.201)(0.801,363.962)(0.851,369.005)(0.901,373.432)(0.951,377.332)(1.001,380.778)

};

\legend{{$\rho=0.1$},{$\rho=0.5$},{$\rho=0.9$}};
\end{axis}
\end{tikzpicture}
}

\caption{\label{fig} Convergence of $M(t)$ towards $M_0$ when $t \to 0$ for $N=m+1=51$. The criterion used is the MSE: $C(t)=E\left[\|M(t)-M_0\|_F^2\right]$.}
\end{center}
\end{figure}
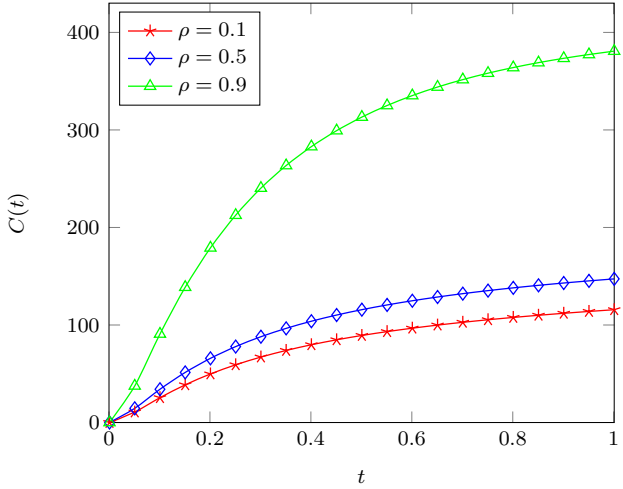

To illustrate Theorem~\ref{th1}, Figure \ref{fig} presents the mean square error $C(t)\triangleq E[\|\Mg(t)-\Mg_0\|_F^2]$ between Tyler's $M$-estimator and the Student-t MLE versus the parameter $t$, called the degree of freedom of the multivariate Student-t distribution \cite{ollila2012complex}, defined through the weight function $u(t,x)=\frac{m+t}{t+x}$. We take here $N=m+1=51$. The data are zero-mean Gaussian distributed with Toeplitz covariance matrix, the $(i, j)$ entry of which is equal to $\rho^{|i-j|}$, for some $\rho \in (0,1)$. As proved in Theorem \ref{th1}, Item (A) is illustrated in the case where $N=m+1$ while Item (B) is illustrated for the Student-t MLE for different population covariance matrices.

\section{Proof of Theorem \ref{th1}}
The strategy of the proof is as follows: for every $t>0$, we first build a positive functional $H(t,\cdot)$ over ${\rm PSD}_m$ whose critical points (if any) are exactly the solutions of $({\rm Eq})_t$. 
To establish the existence of such critical points, we show that $H(t,\cdot)$ is uniformly bounded and tends to zero at the boundary of ${\rm PSD}_m$. 
To obtain uniqueness, we show that solutions of $({\rm Eq})_t$ are all local strict maxima of $H(t,\cdot)$ and conclude by applying the mountain pass theorem (cf. \cite{struwe2008variational}). This gives Item $(A)$. Item $(B)$ is then obtained using the implicit function theorem and some limiting arguments. 

For $t>0$, we define the function
\bea
h:\mR_+^*\times \mR_+&\rightarrow&\mR_+^*\nonumber \\
(t,x)&\mapsto&e^{-\,\frac{1}{m}\int_{x_t}^xu(t,y)dy}.\label{eq:h0}
\eea
Then $-\frac{h_x}h=\frac{u}{m}$ and $h(t,x_t)=1$. Set $h(0,x)=\frac1{x}$ for $x>0$
and $g:\mR_+^*\times \mR_+\rightarrow\mR_+^*$ with
$g(t,x)=xh(t,x)$. 

In the case where $u=\bar{u}$, $\forall (t,x)\in (\mR_+)^2\setminus \{(0,0)\}$, 
\begin{align*}
	x_t\equiv 1,~ \bar{h}(t,x)=\Big(\frac{1+t}{x+t}\Big)^{1+t},~
 \bar{g}(t,x)=x\Big(\frac{1+t}{x+t}\Big)^{1+t}.
\end{align*}

Then, define the functional $H(t,\cdot)$ as
\begin{align}
H:\mR_+^*\times {\rm PSD}_n&\rightarrow\mR_+^*\nonumber \\
(t,\Mg)&\mapsto\frac{\prod_{i=1}^N h(t,\ciMv)^m}{(\det \Mg)^N}\label{eq:H}
\end{align}
as well as the functional considered in \cite{pascal2008covariance}
\begin{align}
B:{\rm PSD}_n&\rightarrow\mR_+^*\nonumber \\
\Mg&\mapsto\frac{\prod_{i=1}^N h(0,\ciMv)^m}{(\det \Mg)^N}\label{eq:H0}.
\end{align}

\begin{lemma}\label{lem1}
For $t>0$ and $\Mg\in {\rm PSD}_m$, one has
$-\Mg H_x(t,\Mg)\Mg/NH(t,\Mg)=\Mg-\frac{1}N\siN u(t,\ciMv)$,
with $H_x(t,\Mg)$ the gradient of $H(t,\cdot)$.
In particular, $\Mg$ is a solution of $({\rm Eq})_t$ if and only if $\Mg$ is a critical point of $H(t,\cdot)$.
\end{lemma}

\begin{lemma}\label{lem2}
$\forall t>0,\Mg\in {\rm PSD}_m$, $H(t,\Mg)\leq B(\Mg)$. As a consequence, $\lim_{\Mg\to \partial {\rm PSD}_m}H(t,\Mg)=0$, so that $H(t,\cdot)$ admits critical points. 
\end{lemma}
\proof[Lemma \ref{lem2}] An immediate calculus yields that $x\mapsto g(t,x)$ reaches its maximum $1$ at $x=x_t$. As a consequence, for $t>0,\Mg\in {\rm PSD}_m$, $H(t,\Mg)\leq B(\Mg)$. Moreover, 
$\lim_{x\rightarrow 0,\iy}g(t,x)=\lim_{x\rightarrow 0,\iy}xh(t,x)=0$.
For the limit at $x=0$, this is obvious. For $x\to\infty$, note that $\ln(g(t,x))=\frac{1}{m}\int_{x_t}^x\frac{m-yu(t,y)}y dy$ and, since $m-l_t<0$, it is equivalent to $(m-l_t)\ln(x)$ as $x\to\infty$.
Consider now a sequence $(\Mg_k)_{k\geq 0}$ in ${\rm PSD}_m$ converging to $\partial {\rm PSD}_m$. For $k\geq 0$, set $\Mg_k=\rho_k\Ng_k$ with
$\rho_k=\Vert \Mg_k\Vert$ and $\Ng_k=\frac{\Mg_k}{\rho_k}$. 
Note that  $\partial {\rm PSD}_m$ is made of matrices either non invertible or with norm going to infinity. Therefore, up to subsequences, either $(i)$ $(\Ng_k)_{k\geq 0}$
converges itself to  $\partial {\rm PSD}_m$ or $(b)$ the sequence $(\rho_k)_{k\geq 0}$ converges to zero or infinity and there exists $\exists \alpha>0,\forall k\geq 0,\Ng_k\geq \al \Ig_m$.
If Case $(i)$ occurs, then $\forall k\geq 0, H(t,\Mg_k)\leq B(\Ng_k)$, which tends to zero as $k\to\infty$ (cf. \cite{pascal2008covariance}). In Case $(ii)$,
\begin{align*}
	H(t,\Mg_k)=\frac{\prod_{i=1}^N h\left(t,x_{i,k}\right)^m}{\rho_k^N\det(\Ng_k)^N}=B(\Ng_k)\prod_{i=1}^N g(t,x_{i,k})^m
\end{align*}
where $x_{i,k}=\yg_i^T\Ng_k^{-1}\yg_i/\rho_k$. As $k\to\infty$, $x_{i,k}$ tends either to zero or infinity and we conclude. 
For $ t>0$, $H(t,\cdot)$ is uniformly bounded over ${\rm PSD}_m$ since $B(\cdot)$ is. So $H(t,\cdot)$ has a global maximum which must belong to ${\rm PSD}_m$ since $H(t,\Mg)\to 0$ as $\Mg$ tends to the boundary of ${\rm PSD}_m$. So $H(t,\cdot)$ admits critical points. 
\EOP

\begin{lemma}\label{lem3}
Let $t>0$. Then all critical points of $H(t,\cdot)$ are local strict maxima.
\end{lemma}
\proof[Lemma \ref{lem3}] We show that, if $\Mg$ is a critical point then the Hessian of $H(t,\cdot)$ at $\Mg$ is a negative definite quadratic form implying that $\Mg$ is a local strict maximum of $H(t,\cdot)$.
Let $\Mg\in {\rm PSD}_m$ be a critical point of $H(t,\cdot)$. Then, one gets that for every $\Qg\in {\rm Sym}_m$,
\begin{align*}
&\langle \Qg,{\rm Hess}_\Mg(\Qg)\rangle=
-NH(t,\Mg)\Big[\langle \Qg,\Mg^{-1}\Qg\Mg^{-1}\rangle \\
&~ ~ ~ ~ ~ ~ ~ ~ ~ +\frac{1}N\siN u_x(t,\ciMv)(\yg_i^T\Mg^{-1}\Qg\Mg^{-1}\yg_i)^2\Big].
\end{align*}
Let $\Rg:=\Mg^{-1/2}\Qg\Mg^{-1/2}$ and $\dg_i:=\Mg^{-1/2}\yg_i$, one has
\begin{align}\label{eq:int0}
&-\frac{\langle \Qg,{\rm Hess}_\Mg(\Qg)\rangle}{NH(t,\Mg)}= \nonumber \\
&~ ~ ~ ~ ~ ~ ~ ~ ~ ~ ~ ~ ~ ~ \Vert \Rg\Vert^2+\frac{1}N\siN u_x(t,\Vert \dg_i\Vert^2)(\dg_i^T\Rg\dg_i)^2.
\end{align}
Recall that $\Mg$ is a critical point of $H(t,\cdot)$ and thus a solution of $({\rm Eq})_t$, i.e.,
\begin{equation}\label{eq:id0}
	\Ig_m=\frac{1}N\siN u(t,\Vert \dg_i\Vert^2)\dii.
\end{equation}
Multiplying \eqref{eq:id0} by $\Rg$ on both left and right, taking the trace and plugging the result into \eqref{eq:int0} gives
\begin{align*}
	&\eqref{eq:int0}=\frac{1}N\siN u(t,\Vert \dg_i\Vert^2)\Vert \Rg\dg_i\Vert^2+u_x(t,\Vert \dg_i\Vert^2)(\dg_i^T\Rg\dg_i)^2.
\end{align*}
Let $I_\Qg=\{i\in \{1,\cdots,N\}, \Rg\dg_i\neq 0\}$. Then
\begin{align*}
	&\eqref{eq:int0}=\frac{1}N\sum_{i\in I_\Qg}\Vert \Rg\dg_i\Vert^2\left[u(t,\Vert \dg_i\Vert^2)+\Vert \dg_i\Vert^2u_x(t,\Vert \dg_i\Vert^2)r_i\right]
\end{align*}
where $r_i:=(\dg_i^T\Rg\dg_i /[{\Vert \dg_i\Vert\Vert \Rg\dg_i\Vert}])^2$. Using $0\leq r_i\leq 1$ (by Cauchy-Schwarz's inequality) and $u_x\leq 0$ (since $u$ is of class $C^1$ and verifies $(U1)$), we have $r_iu_x(\cdot,\cdot)\geq u_x(\cdot,\cdot)$. Then, recalling that $v(t,x)=xu(t,x)$,
\begin{align*}
	&\eqref{eq:int0} \geq \frac{1}N\sum_{i\in I_\Qg}\Vert \Rg\dg_i\Vert^2v_x(t,\Vert \dg_i\Vert^2)\geq 0.
\end{align*}
Moreover, if $\Qg\neq 0$, $I_\Qg\neq \emptyset$ and there exists $\bar{i}$ such that $v_x(t,\Vert \dg_{\bar{i}}\Vert^2)>0$. Therefore $\langle \Qg,{\rm Hess}_\Mg(\Qg)<0$, i.e., ${\rm Hess}_\Mg$ is negative definite, concluding the proof.
\EOP
\begin{lemma}\label{lem4}
	Let $t>0$. Then $({\rm Eq})_t$ admits a unique solution, $\Mg(t)$, the unique strict maximum of $H(t,\cdot)$. 
\end{lemma}
\proof[Lemma \ref{lem4}] 
We reason by contradiction assuming $H(t,\cdot)$ admits at least two local strict maxima. Applying the mountain-pass theorem \cite{struwe2008variational} to the functional $1/H(t,\cdot)$ which tends to infinity in the vicinity of $\partial{\rm PSD}_m$, we obtain the existence of a saddle point of $\Fg$ in ${\rm PSD}_m$ which is contradictory to Lemma~\ref{lem3}.
\EOP\\

We next prove that $\Mg(t)$ is uniformly bounded in ${\rm PSD}_m$ as $t\to 0$, i.e.
\begin{lemma}\label{lem5}
There exists $0<a\leq b$ and $t_0>0$ such that, for every $t\in (0,t_0)$, $aI_m\leq \Mg(t)\leq b\Ig_m$. 
\end{lemma}
\proof[Lemma \ref{lem5}] Let $\Pg$ be the unique matrix of ${\rm PSD}_m$ satisfying $B(\Pg)=\max_{\Mg\in {\rm PSD}_m}B(\Mg)$ and $\Tr(\Mg)=m$. Then, for every $t>0$, $H(t,\Pg)\leq H(t,\Mg(t))$ and $B(\Mg(t))\leq B(\Pg)$. Multiplying both inequalities, after simplifications, we get
	$\prod_{i=1}^N g(t,\yg_i^T\Pg^{-1}\yg_i)\leq \prod_{i=1}^N g(t,\yg_i^T\Mg(t)^{-1}\yg_i)\leq 1$,
with $\prod_{i=1}^N g(t,\yg_i^T\Pg^{-1}\yg_i)\to 1$ as $t\to 0$. So there exists $t_0>0$ such that, for every $t\in (0,t_0)$ and $1\leq i\leq N$,
$1/2\leq g(t,\yg_i^T\Mg(t)^{-1}\yg_i)$,
and, since $(U3)$ holds true, there exists $0<a\leq b$ s.t. for every $t\in (0,t_0)$ and $1\leq i\leq N$, $a\leq \yg_i^T\Mg(t)^{-1}\yg_i\leq b$. This implies that, for every $t\in (0,t_0)$ and $1\leq i\leq N$, $u(t,b)\leq u(t,\yg_i^T\Mg(t)^{-1}\yg_i))\leq u(t,a)$, hence $u(t,b)\Cg\leq \Mg(t)\leq u(t,a)\Cg$ with 
$\Cg:=\frac{m}N\siN\cii$. One concludes easily.
 \EOP

\begin{lemma}\label{lem6}
Under the conditions of Theorem~\ref{th1}, $\lim_{t\to 0}\Mg(t)=\Mg_0$ solution of $({\rm Eq})_0$ given by $\Mg_0=\xi_u \Pg$, where
$\xi_u>0$ is the unique solution of \eqref{eq:xi0}.
\end{lemma}
\proof[Lemma \ref{lem6}] Since $\Mg(\cdot)$ is uniformly bounded in ${\rm PSD}_m$ as $t\to 0$, its accumulation points still belong to ${\rm PSD}_m$ and are necessarily of the form $\mu \Pg$ where $\mu>0$ and $\Pg$ is the solution of $({\rm Eq})_0$ with trace $m$. Taking the trace in \eqref{eq:id0}, one gets $m=\frac{1}N\siN v(t,\Vert \dg_{i}(t)\Vert^2)$,
where $\dg_{i}(t)=\Mg(t)^{-1/2}\yg_i$ for $1\leq i\leq N$. Using \eqref{eq:v0} and $(U3)$, one deduces that, for every $t>0$,
$\siN v_1(\Vert \dg_{i}(t)\Vert^2)+t\siN w_1(\Vert \dg_{i}(t)\Vert^2)+o(t)=0$.
Consider an accumulation point $\mu \Pg$ of $\Mg(\cdot)$ as $t\to 0$. Then, up to a subsequence, $\lim_{t\to 0}\Mg(t)=\mu \Pg$ and, for $1\leq i\leq N$,
$\lim_{t\to 0}\dg_{i}(t)={\Pg^{-1/2}\yg_i}/{\sqrt{\mu}}$. According to $(U3)$, the second sum in the previous equation tends to zero as $t\to 0$ and we are left with $\siN v_1({\yg_i^T\Pg^{-1}\yg_i}/{\mu})=0$. Since the left-hand side of the latter defines a decreasing function of $\mu$, it has a unique solution denoted $\xi_u>0$, which concludes the proof since $\Mg(\cdot)$ admits a unique accumulation point as $t\to 0$.
\EOP


\section{Conclusions}
{In this paper, an alternative proof for existence and uniqueness for the Maronna's $M$-estimators is provided. More importantly, using this particular approach leads to draw some connections between Maronna's and Tyler's estimators by expressing (properly scaled) Tyler's estimator in terms of a limit of a class of Maronna's estimators. This result may also find interest in studies of Tyler's $M$-estimator in the large random matrix regime.}

\bibliographystyle{IEEEbib}
\bibliography{biblio}

\end{document}